\documentclass[envcountsect, fleqn, 11pt]{llncs}

\usepackage[margin=1in]{geometry}
\usepackage{pdfpages}
\usepackage{amssymb}
\usepackage{amsmath}
\usepackage{mathtools}
\usepackage{graphicx}
\usepackage{xspace}
\usepackage{color}

\usepackage[ruled,vlined,linesnumbered]{algorithm2e}
\usepackage{algpseudocode}
\usepackage{caption}

\begin{document}

\title{Budget Feasible Mechanisms on Matroids}
\author{Stefano Leonardi \inst{1} \and Gianpiero Monaco \inst{2} \and Piotr Sankowski \inst{3} \and Qiang Zhang \inst{1}}
\institute{Sapienza University of Rome, Italy, \email{leonardi@dis.uniroma1.it}, \email{qzhang@gmail.com}
\and University of L'Aquila, Italy, \email{gianpiero.monaco@univaq.it} \and University of Warsaw, Poland, \email{sank@mimuw.edu.pl}}

\maketitle

\begin{abstract}
Motivated by many practical applications, in this paper we study {\em budget feasible mechanisms} where the
	goal is to procure independent sets from matroids. More
	specifically, we are given a matroid $\mathcal{M}=(E,\mathcal{I})$
	where each ground (indivisible) element is a selfish agent. The cost of each element (i.e., for selling the item or performing a service) is only known to the
	element itself. There is a buyer with a budget having additive valuations over the set of elements $E$. The goal is to design an incentive compatible (truthful) 
budget feasible mechanism which procures an independent set of the matroid under the given budget that yields the largest value possible to the buyer. Our result is a deterministic, polynomial-time, 
individually rational, truthful and budget feasible mechanism with
	$4$-approximation to the optimal independent set. Then, we extend our mechanism to the setting of matroid intersections in which the goal is to
	procure common independent sets from multiple matroids. We show that, given a polynomial time deterministic blackbox that returns $\alpha-$approximation solutions to the matroid intersection problem, there exists a deterministic, polynomial time, individually rational, truthful and budget feasible mechanism with $(3\alpha +1)-$approximation to the optimal common independent set.	
\end{abstract}


\section{Introduction}
Procurement auctions (a.k.a. reverse auctions), often carried out by governments or private companies, deal with the scenarios where a buyer would like to purchase objects from a set of sellers. These objects are not limited to physical items. For instance they can be services provided by sellers. 
In this work we consider the problem where a buyer with a budget is interested in a set of indivisible objects for which he has additive valuations. We assume that each object is a selfish agent. More specifically, we assume agents have quasi-linear utilities and they are rational (i.e.,  they aim to maximize the differences between the payments they receive and their true costs). We also restrict ourself to the case where the buyer is constrained to purchase a subset of objects that forms an independent set with respect to an underlying matroid structure. A wide variety of research studies have shown that matroids are linked to many interesting applications, for example, auctions~\cite{ausubel2004efficient,goel2015polyhedral,kleinberg2012matroid}, spectrum market~\cite{tse1998multiaccess}, scheduling matroids~\cite{demange1986multi} and house market~\cite{zhang2016house}. 

One challenge in such procurement auctions involves providing incentives to sellers for declaring their true costs when those costs are their {\em private information}. A classical mechanism, namely Vickrey-Clark-Groves (VCG)
mechanism~\cite{clarke1971multipart,groves1973incentives,vickrey1961counterspeculation}, provides an intuitive 
solution to this problem. The VCG mechanism returns a procurement
that maximizes the valuation of the buyer and the payments for sellers
 are their externalities to the procurement.
The VCG mechanism is a {\em truthful} mechanism, i.e., no seller will improve its utility by manipulating its cost regardless the costs declared by others.
However, the VCG mechanism also has its drawbacks. One of the
drawbacks, which makes VCG mechanism impractical, is that the payments
to sellers could be very high. To overcome this problem two different 
approaches have been proposed and
investigated. The first one is studying the {\em frugality} of
mechanisms~\cite{karlin2005beyond}, which studies the
minimum payment the buyer needs to pay for a set of objects when
sellers are rational utility maximizers. The other approach is developing  
{\em budget feasible mechanisms}~\cite{singer2010budget}, 
where the goal is to maximize the buyer's value for the procurement
under a given budget when
sellers are rational utility maximizers.
 Singer~\cite{singer2010budget} showed that budget
feasible mechanisms could approximate the optimal
procurement that ``magically" knows the costs of sellers, when the buyer's
valuation is nondecreasing submodular.

\noindent{\bf Our Results.}
The goal of this study is to design budget feasible mechanisms for
procuring objects that form an independent set in a given matroid structure. 
To the best of our knowledge it is the first time that matroid constraints are considered in the budget feasible mechanisms setting examined here. Previous work was mainly devoted to different
types of valuations for the buyer (see the Related Work subsection). 
Our results are positive. In Section~\ref{sec_Mechanisms for Matroids} we give a deterministic, polynomial time, 
individually rational, truthful and budget feasible
mechanism with $4$-approximation to the optimal
independent set (i.e., the independent set with maximum value for the buyer under the given budget) within the budget of the buyer  when the buyer has additive valuations. To generalize this result we also provide a
similar mechanism to procure the intersection of independent sets in multiple matroids.
In particular, given a deterministic polynomial time $\alpha$-approximation algorithm for the matroid intersection problems as a blackbox, in Section~\ref{sec_Mechanisms for matroid intersections} we present a deterministic, polynomial time, individually rational, truthful and budget feasible
 mechanism with $(3\alpha+1)$-approximation to the optimal
independent set within the budget of the buyer when the buyer has additive valuations. 
It is also good to know the limitations (e.g. lower bounds) of such
budget feasible mechanisms. In particular the lower bound to any
deterministic mechanism of $1+\sqrt{2}$ for additive valuations with one
buyer presented in~\cite{chen2011approximability} (it is worth
noticing that such lower bound do not rely on any computational or
complexity assumption), suggests that our mechanisms are not far
away from the optimal ones. 
Finally, budget feasible mechanisms also received a lot of attention  when the valuation functions are submodular~\cite{singer2010budget} and XOS~\cite{bei2012budget}. In Section~\ref{sec_XOS functions} we slightly improve the analysis of the mechanism proposed in~\cite{bei2012budget}. Specifically, we improve the approximation ratio of the mechanism from $768$ to $436$ by tuning the parameters in the mechanism.


\noindent{\bf Related Work.} The study of budget feasible mechanisms was
initiated in~\cite{singer2010budget}. It essentially focuses on the
procurement auctions when sellers have private costs for their
objects and a buyer aims to maximize his valuation function on
subsets of objects, conditioned on that the sum of the payments
given to sellers {\em cannot} exceed a given budget of the buyer. In
particular Singer~\cite{singer2010budget} considered budget
feasible mechanisms when the valuation function of the buyer is
nondecreasing submodular. For general nondecreasing submodular
functions, Singer~\cite{singer2010budget} gave a lower bound of $2$
for deterministic budget feasible mechanisms and a randomized budget
feasible mechanism with $112$-approximation. When the valuation function of the buyer is additive, a special class of
nondecreasing submodular functions, 
Singer~\cite{singer2010budget} gave a polynomial deterministic
budget feasible mechanism with $6$-approximation and a lower bound
of $2$ for any deterministic budget feasible mechanism. All results
were improved in~\cite{chen2011approximability}, for example, a
deterministic budget feasible mechanism with
$2+\sqrt{2}$-approximation and an improved lower bound of
$1+\sqrt{2}$ for any deterministic budget feasible mechanism for
additive valuations were given.  Furthermore, Bei et
al.~\cite{bei2012budget} gave a 768-approximation mechanism for XOS valuations and extended their study to Bayesian settings. Chan and Chen~\cite{chan2014truthful} studied  budget feasible mechanisms in the settings in which each seller processes multiple copies of the objects. They gave logarithmic mechanisms for concave additive valuations and sub-additive valuations. 

Budget feasible mechanisms are attractive to many
communities due to their various
applications. In crowdsourcing the goal is to
assign skilled workers to tasks when workers have private costs.  By injecting some
characteristics in crowdsourcing, budget feasible mechanisms have
been further developed and improved. For example, Goel et
al.~\cite{goel2014allocating} developed budget feasible mechanisms
that achieve $\frac{2e-1}{e-1}$-approximation to the optimal social
welfare by exploiting the assumption that one worker has limited
contribution to the social welfare. Furthermore
Anari et al.~\cite{anari2014mechanism} gave a budget feasible
mechanism that achieves a competitive ratio of $1-1/e \approx 0.63$
by using the assumption that the cost of any worker is relatively
small compared to the budget of the buyer. 

Another work close to ours is~\cite{bikhchandani2011ascending}, which studies the ``dual" problem of maximizing the revenue by selling the maximum independent set of a matroid. They proposed a truthful ascending auction in which  a seller is constrained to sell objects that forms a basis in a matroid. 


\section{Preliminaries}
\label{sec:preliminaries}
\paragraph{Matroids.} A matroid $\mathcal{M}$ is a pair of $(E, \mathcal{I})$ where $E$ is a ground set of finite elements and $\mathcal{I} \subseteq 2^{E}$ consists of subsets of the ground set satisfying the following properties:

\begin{itemize}
	\item Hereditary property: If $I \in \mathcal{I}$, then $J \in \mathcal{I}$ for every $J \subset I$.
	\item Exchange property: For any pair of sets $I, J \in \mathcal{I}$, if $|I| < |J|$, then there exists an element $e \in J$ such that $I \cup \{e\} \in \mathcal{I}$.
\end{itemize}

The sets in $\mathcal{I}$ are called \textit{independent sets}.
Given a matroid $\mathcal{M}=(E, \mathcal{I})$ and $T \subseteq E$ is a subset of $E$, the {\em restriction} of $\mathcal{M}$ to $T$, denoted by $\mathcal{M}|T$, is the matroid in which the ground set is $T$ and the independent sets are the independent sets of $\mathcal{M}$ that are contained in $T$. That is, $\mathcal{M} | T = (T, \mathcal{I}(\mathcal{M}|T))$ where $\mathcal{I}(\mathcal{M}|T) = \{ I \subseteq T : I \in \mathcal{I} \}$. Similarly,  the {\em deletion} of $\mathcal{M}$, denoted by $M\setminus T$, is the matroid in which the ground set is $E - T$ and the independent sets are the independent sets of $\mathcal{M}$ that do not contain any element in $T$. That is, $\mathcal{M} \setminus T = (E - T,  \mathcal{I}(\mathcal{M} \setminus T))$ where $\mathcal{I}( \mathcal{M} \setminus T) = \{I \subseteq E - T: I \in \mathcal{I}\}$.

\paragraph{Matroid Budget Feasible Mechanisms.} In an instance of the matroid budget feasible mechanism design problem, we are given a matroid $\mathcal{M} = (E, \mathcal{I})$ consisting of $n$ ground elements, each of whom is associated with a weight $w_e \in \mathbb{R}_+$. Each element $e \in {E}$ is also associated with a private cost $c_e \in \mathbb{R}_+$, which is only known to the element itself. Our goal is to design a truthful mechanism that gives incentives to elements for declaring their private costs truthfully and then selects an independent set conditioned on that the total payment given to the elements does not exceed a given budget $b$.
Given an independent set $I \in \mathcal{I}$, the value of the independent set is defined by $w(I) = \sum_{e \in I} w(e)$. We compare the value of the independent set selected by the
mechanism against the value of the maximum-value independent set in which the total true cost of elements does not exceed the budget.

We use $\mathbf{w} = \langle w_1, \ldots, w_n \rangle$ to denote the weight of the ground elements and use $\mathbf{d}=\langle d_1,\ldots,d_n\rangle$  to denote the costs declared by the ground elements.
Let $\tau$ be the maximum-weight element (breaking ties arbitrarily), that is, $w_\tau = \max_{e \in E}w_e$.  We assume that $d_e \in \mathbb{R}_+$ and $d_e \leq b$ for any $e \in E$ since elements with costs greater than $b$ cannot be selected by any mechanism due to the budget constraint. This also implies that no element could improve its utility by declaring $d_i > b$. Given a subset of element $T$, we use $\mathbf{w}_{-T}$ and $\mathbf{d}_{-T}$ to denote the weight and cost vector excluding elements in $T$. Similarly, we use  $\mathbf{w}_{T}$ and $\mathbf{d}_{T}$ to denote the weight and cost vector only including elements in $T$.  For each element $e \in {E}$, $\mathsf{bb}(e) = \frac{d_e}{w_e}$ is called the {\em buck-per-bang} rate for element $e$.\footnote{$\frac{w_e}{c_e}$ is usually known as the \textit{bang-per-buck} rate. To simplify the presentation, we call $\frac{d_e}{w_e}$ the {\em buck-per-bang} rate.}

A deterministic mechanism $M = (f,p)$
consists of an allocation function $f:  \mathcal{M}, \mathbf{w}, \mathbf{d}, b
\to I \in \mathcal{I}$ and a payment function $p: \mathcal{M},\mathbf{w}, \mathbf{d}, b
\to \mathbb{R}_+^n$. Given the weights and declared costs of the ground elements, the allocation function returns an independent set in the matroid and the payment function indicates the payments for all elements.
Let $\mathbf{f}^{{M}}(\mathcal{M},\mathbf{w},
\mathbf{d},b)$ and $\mathbf{p}^{{M}}(\mathcal{M}, \mathbf{w},
\mathbf{d}, b)$ be the independent set and payments returned by
${M}$, respectively. If element $e$ is in the independent set
obtained by ${M}$, then $f^{M}_e(\mathcal{M}, \mathbf{w},
\mathbf{d},b)=1$. Otherwise, $f^{M}_e(\mathcal{M}, \mathbf{w},
\mathbf{d},b)=0$. It is assumed that $p^{M}_e(\mathcal{M}, \mathbf{w},
\mathbf{d},b)=0$ if $f^{M}_e(\mathcal{M}, \mathbf{w}, \mathbf{d},b)=0$.
The utility of an element is the difference between the payment received from the mechanism and its true cost. More specifically, the utility of element $e$ is given by $u^{{M}}_e(\mathcal{M},
\mathbf{w}, \mathbf{d},b) =   p^{M}_e(\mathcal{M}, \mathbf{w},
\mathbf{d},b) - f^{M}_e(\mathcal{M}, \mathbf{w}, \mathbf{d},b) \cdot c_e$.

\paragraph{Individual Rationality:}
A mechanism $M$ is {\em individually rational} if $p^M_e(\mathcal{M}, \mathbf{w}, \mathbf{d},b) - f^{M}_e(\mathcal{M},\mathbf{w}, \mathbf{d},b) \cdot d_e \geq 0$ for any $\mathcal{M}$, any $\mathbf{w} \in \mathbb{R}_+^n$, any $\mathbf{d} \in \mathbb{R}_+^n$,  any $b \in \mathbb{R}_+$ and any element $e \in E$. That is, no element in the selected independent set is paid less than the cost it declared.

\paragraph{Truthfulness:}
A mechanism ${M}$ is {\em truthful} if it holds  $ u^{{M}}_e(\mathcal{M}, \mathbf{w}, \mathbf{d}_{-e}, c_e, b) \geq u^{{M}}_e(\mathcal{M}, \mathbf{w},\mathbf{d}_{-e}, d_e, b)$ for any $\mathcal{M}$, any $\mathbf{w} \in \mathbb{R}_+^n$, any $\mathbf{d}_{-e} \in \mathbb{R}_+^{n-1}$, any $d_e \in \mathbb{R}_+$, any $c_{e} \in \mathbb{R}_+$,  $b \in \mathbb{R}_+$ and any $e  \in E$, where $\mathbf{d}_{-e} = \langle d_1,\ldots,d_{e-1},d_{e+1},\ldots, d_{n} \rangle$. When the context is clear, we sometimes abuse some notations. For example, here we write  $u^{{M}}_e(\mathcal{M}, \mathbf{w}, \mathbf{d}_{-e}, c_e,b)$ instead of $u^{{M}}_e(\mathcal{M}, \mathbf{w}, \langle\mathbf{d}_{-e}, c_e\rangle,b)$.  A truthful mechanism prevents any element improving its utility by mis-declaring its cost regardless the costs declared by other elements.

\paragraph{Budget Feasibility:}
A mechanism ${M}$ is {\em budget feasible} if it holds that $\sum_{e \in {E}} p^{{M}}_e(\mathcal{M},\mathbf{w}, \mathbf{d},b) \leq b$ for any $\mathcal{M}, \mathbf{w}\in \mathbb{R}_+^n$,  any $\mathbf{d} \in \mathbb{R}_+^n$ and any $b \in \mathbb{R}_+$.

\paragraph{Competitiveness:}
A mechanism ${M}$ is {\em $\alpha$-competitive} if   $w(f^{{M}}(\mathcal{M},\mathbf{w},\mathbf{d},b)) \geq \frac{1}{\alpha} w( \mathsf{OPT}(\mathcal{M}, \mathbf{w},\mathbf{d},b))$  for any $\mathbf{w} \in \mathbb{R}_+^n, \mathbf{d} \in \mathbb{R}_+^n$ and $b \in \mathbb{R}_+$, where $\mathsf{OPT}(\mathcal{M}, \mathbf{w},\mathbf{d},b)$ is the maximum-value independent set in which the total cost of the elements is at most $b$. We often call  $\mathsf{OPT}(\mathcal{M}, \mathbf{w},\mathbf{d},b)$ the optimal independent set
and simplify it as $\mathsf{OPT}(\mathcal{M},b)$ throughout the paper when the weights and the costs of elements are clear. Similarly we use $\mathsf{MAX}(\mathcal{M},\mathbf{w})$, shorten by $\mathsf{MAX}(\mathcal{M})$,  to denote the maximum-value independent set in $\mathcal{M}$ without considering the budget constraint.

\paragraph{Simplifying notations.}  From now on to avoid heavy notations we sometimes simplify the notations. For example we will write $f^{M}, f^{M}_e, p^{M}, p^{M}_e$ when the inputs of the mechanism are clear. And we will use $\mathsf{OPT}(\mathcal{M}\setminus T,b)$ instead of $\mathsf{OPT}(\mathcal{M}\setminus T, \mathbf{w}_{-T}, \mathbf{d}_{-T}, b)$ to denote the optimal independent set in matroid $\mathcal{M}\setminus T$. Similarly  we will use $\mathsf{OPT}(\mathcal{M} |T,b)$ instead of $\mathsf{OPT}(\mathcal{M} | T, \mathbf{w}_{T}, \mathbf{d}_T, b)$ to denote the optimal independent set in matroid $M | T$. Furthermore we use $\mathsf{MAX}(\mathcal{M} \setminus T)$ instead of  $\mathsf{MAX}(\mathcal{M} \setminus T, \mathbf{w}_{-T})$ to denote  the maximum-value independent set in $\mathcal{M}\setminus T$ without considering the costs of the elements and the budget.

\section{Mechanisms for Matroids}\label{sec_Mechanisms for Matroids}
In this section we provide our main result. We give a deterministic, polynomial time, individually rational, truthful and budget feasible mechanism that is 4-approximating the optimal independent set. Before providing the mechanism we discuss some intuition that guides us in the design of Mechanism~\ref{alg:budgetedMatroids}. First imagine that there exists an element with a very high weight, i.e., any independent set without this element results in a poor value compared to the optimal independent set. In this case that element may strategically declare a high cost in order to increase its utility as it knows that any competitive mechanism has to select it. To avoid that this happens we remove element $\tau$ (i.e., the element with the largest weight) from the matroid via matroid deletion operation, and compare it with the independent set computed later by the mechanism. Second we observe that most of the existing budget feasible mechanisms adopt proportional payment schemes, where elements (i.e., agents) are paid proportionally according to their contribution in the solution.  In other words in a proportional payment scheme there is an uniform price such that the payments for elements in the solution are the products of their contribution and this price.  In addition greedy algorithms are commonly used in matroid systems. Combining these two observations our plan is to start from a high price and compute the maximum-value independent set in the matroid at each iteration. If there is enough budget to pay this independent set at the current price then we proceed to the final step of the mechanism. Otherwise we reduce the price and remove an element from the matroid. The buck-per-bang rate of that element becomes an upper bound of the payment on each contribution in the next iteration. The mechanism performs the procedure described above until the payment of the maximum-value independent set is within budget $b$.
As we will show next, if the value of the optimal independent set does not come from a single element, we are able to retain most of the value of the optimal independent set after removing those elements. Finally, we show that returning the better solution between the maximum-value independent set found and element $\tau$ approximates the value of the optimal independent set within a factor of $4$.

\begin{algorithm}[h]
	\KwIn{$\mathcal{M}=(E,\mathcal{I}), \mathbf{w}, \mathbf{d},b$}
	\KwOut{$\mathbf{f},\mathbf{p}$}
	Sort elements in $E - \tau$ in a non-increasing order of buck per bang, i.e. $\mathsf{bb}(i) \geq \mathsf{bb}(j)$ if $ i < j$, break ties arbitrarily\;
	Let $\mathsf{bb}(0) = +\infty$, $i=1$ and $T = \emptyset$\;
	Set $r = \mathsf{bb}(i)$\;
	\While{$w(\mathsf{MAX}(\mathcal{M}\setminus (T \cup \tau) ))\cdot r > b$}{
			$T = T \cup \{i\}$ and $i = i+1$\;
	}
	$r = \min\{ \frac{b}{w(\mathsf{MAX}(\mathcal{M}\setminus (T \cup \tau) ))},  \mathsf{bb}(i-1)\}$\label{line:HatRBudget}\;
\uIf{$w(\mathsf{MAX}(\mathcal{M}\setminus (T \cup \tau) ))>w_\tau$}{
	For each $e \in E$, if $e \in \mathsf{MAX}(\mathcal{M}\setminus (T \cup \tau) ), f_e = 1$ and $ p_e = r \cdot w_e$. Otherwise, $f_e = 0$ and $ p_e = 0$\;
}\Else{
$f_\tau = 1, p_\tau = b$. For edge $e \in E - \tau, f_e = 0, p_e = 0$\;
}
\KwRet{$\mathbf{f}, \mathbf{p}$}\;
\SetAlgorithmName{Mechanism}{mechanism}{List of Mechanisms}
\caption{A budget feasible mechanism for procuring independent sets in matroids}
\label{alg:budgetedMatroids}
\end{algorithm}
\begin{theorem}
	Mechanism~\ref{alg:budgetedMatroids} is a deterministic, polynomial time, individually rational, truthful and budget feasible  mechanism that is $4$-competitive against the optimal independent set given a budget. 
\end{theorem}

\subsection{Approximation}
Recall that ${T}$ is the set of elements removed from the matroid.
$\mathsf{MAX}(\mathcal{M}\setminus (T \cup \tau) )$ is the independent set found when Mechanism~\ref{alg:budgetedMatroids} stops, and it is also the maximal-value independent set in matroid $M\setminus (T \cup \tau)$. The roadmap of the proof is to first show that,  the independent set $\mathsf{MAX}(\mathcal{M}\setminus (T \cup \tau) )$
well approximates the optimal independent set in matroid $M \setminus \tau$. Next we show
that returning the maximum between $\tau$ and $\mathsf{MAX}(\mathcal{M}\setminus (T \cup \tau) )$
gives $4$-approximation to the optimal independent set in matroid $\mathcal{M}$.

\begin{lemma}
	\label{lem:approximaiton}
	Given any $\mathcal{M}, \mathbf{w}, \mathbf{d}, b$, when Mechanism~\ref{alg:budgetedMatroids} stops, it holds
	\[
	w(\mathsf{OPT}(\mathcal{M} \setminus  \tau, b)) \leq 2 w(\mathsf{MAX}(\mathcal{M}\setminus (T \cup \tau) )) + w_\tau
	\]
\end{lemma}
\begin{proof}
	It is trivial to see that this lemma holds when $\tau$ is the only element in matroid $\mathcal{M}$.
	The rest of the proof uses a similar idea in~\cite{goel2014allocating} and is divided into two cases depending on whether the full budget $b$ is spent or not. 
	Consider $E - \{\tau\} $ is partitioned into two disjoint sets, $E - \{\tau\} - T$ and ${T}$.
	The value of maximum-value independent set $w(\mathsf{OPT}(\mathcal{M} \setminus  \tau, b)) $ is bounded by
	\[
	w(\mathsf{OPT}(\mathcal{M} | T, b)) +  w(\mathsf{OPT}(\mathcal{M} \setminus (T \cup \tau)    ,b))
	\]
	As the buck-per-bang is at least $r$ for every element in $T$, the optimal independent set given a budget $b$ in $\mathcal{M} | T$, i.e. $w(\mathsf{OPT}(\mathcal{M} | T, b)) $,  is at most $b/r$. When the full budget is spent, the independent set $f^M$ is $b/r$ in Mechanism~\ref{alg:budgetedMatroids}. On the other hand, $f^M$ is the maximum-value independent set in $\mathcal{M} \setminus (T \cup \tau)$. It implies that $w(\mathsf{MAX}(\mathcal{M}\setminus (T \cup \tau) )) \geq w(\mathsf{OPT}(\mathcal{M} \setminus (T\cup \tau),b))$. The above analysis concludes that
	\[
		w(\mathsf{OPT}(\mathcal{M} \setminus  \tau, b))  \leq 2 w(\mathsf{MAX}(\mathcal{M}\setminus (T \cup \tau) ))
	\]
	
	Now we turn to the case that some budget is left in Mechanism~\ref{alg:budgetedMatroids}. Note that it happens because $r = \mathsf{bb}(i-1)$ (see Line~\ref{line:HatRBudget}) during the execution of Mechanism~\ref{alg:budgetedMatroids}. Since Mechanism~\ref{alg:budgetedMatroids} does not stop when $r =\mathsf{bb}(i-1)$, it implies that the maximum-value independent set found was not budget feasible at previous iteration. 
	 After removing element $i-1$, the maximum-value independent set becomes budget feasible when $r= \mathsf{bb}(i-1)$. These together imply
	\[
	w(\mathsf{MAX}(\mathcal{M} \setminus (T' \cup \tau) )) \cdot \mathsf{bb}(i-1)  > b > w(\mathsf{MAX}(\mathcal{M} \setminus (T \cup \tau) ))  \cdot \mathsf{bb}(i-1)
	\]
	where $T' = T - \{i-1\}$.
	This further implies that budget left is at most $\mathsf{bb}({i-1})\cdot w_{i-1}$. By the similar argument as in previous case, the optimal independent set in $\mathcal{M}|T$ is at most $b/r$, while the value of the independent set $\mathsf{MAX}(\mathcal{M}\setminus (T \cup \tau) )$
	 is at least $(b - \mathsf{bb}(i-1)\cdot w_{i-1})/r$, which is at least $b/r - w_{i-1}$.
	 Therefore, we have
	\[
		w(\mathsf{OPT}(\mathcal{M} \setminus  \tau, b)) \leq 2 w(\mathsf{MAX}(\mathcal{M}\setminus (T \cup \tau) )) +  w_{i-1}
	\]
	Substituting $w_{i-1}$ with $w_\tau$ completes the proof.
\qed
\end{proof}

Next, we show that returning the maximum between $\tau$ and $\mathsf{MAX}(\mathcal{M}\setminus (T \cup \tau) )$ is $4-$competitive against the optimal independent set in $\mathcal{M}$.

\begin{lemma}
	Given any $\mathcal{M}, \mathbf{w}, \mathbf{d}, b$, the independent set returned by Mechanism~\ref{alg:budgetedMatroids}, i.e., the maximum between $\tau$ or $\mathsf{MAX}(\mathcal{M}\setminus (T \cup \tau) )$, is $4-$competitive against the optimal independent set.
\end{lemma}
\begin{proof}
	The optimal independent set in $\mathcal{M}$ is bounded by
	\[
	w(\mathsf{OPT}(\mathcal{M},b)) \leq w_\tau + w(\mathsf{OPT}(\mathcal{M} \setminus  \tau, b))
	\]
	By Lemma~\ref{lem:approximaiton}, we have
	\[
	w(\mathsf{OPT}(\mathcal{M},b)) \leq 2w_\tau + 2w(\mathsf{MAX}(\mathcal{M}\setminus (T \cup \tau) ))
	\]
	Therefore, the maximum between $\tau$ and $\mathsf{MAX}(\mathcal{M}\setminus (T \cup \tau) )$ approximates the optimal independent set within a factor of 4.
\qed
\end{proof}

\subsection{Truthfulness}
\label{sec:truthfulmatroid}
In this section, we will show that Mechanism~\ref{alg:budgetedMatroids} is  truthful. It is easy to see that  element $\tau$ cannot benefit by manipulating its cost.

\begin{lemma}
	The element with the maximum weight, i.e., element $\tau$, could not improve his utility by declaring cost $d_\tau \neq c_\tau$.
\end{lemma}
\begin{proof}
	If Mechanism~\ref{alg:budgetedMatroids} returns element $\tau$ when  $\tau$ declares his true cost, then $\tau$ gets a payment of $b$ so that there is no incentive for him to declare other cost. On the other hand, when Mechanism~\ref{alg:budgetedMatroids} returns $\mathsf{MAX}(\mathcal{M}\setminus (T\cup \tau))$, declaring a different cost will not change the outcome as it is still the element with the largest weight and $w_\tau < w(\mathsf{MAX}(\mathcal{M}\setminus (T\cup \tau)))$.
\qed
\end{proof}

Next we show that no edge in $E - \tau$ could improve his utility by mis-declaring its cost. The proof relies on the analysis of different cases. The first case shows that those removed element in  Mechanism~\ref{alg:budgetedMatroids} cannot benefit by manipulating their costs.

\begin{lemma}
\label{lem:truthful1}
		Assume an element $k$ in  $T$ when it declares its cost truthfully. Then, element $k$ could not improve his utility by declaring a cost $d_k \neq c_k$.
\end{lemma}
\begin{proof}
	As $k \in T$, we know that element $k$ is not in the independent set returned by Mechanism~\ref{alg:budgetedMatroids}. Hence, his utility is zero. It
	implies, when element $k$ is considered, that is, $r = \mathsf{bb}(k)$, it holds that $w(\mathsf{MAX}(\mathcal{M} \setminus (T^{k} \cup \tau)  )  )) \cdot r > b$ where $T^k$ denotes the set of elements removed until $k$ is considered. Consider that element $k$ declares a higher cost $d_k > c_k$ and it is considered earlier at the $h^{th}$ iteration where $h\leq k$. Equivalently speaking,  $k$ becomes the element with the $h^{th}$ largest buck-per-bang rate.  In this case,  Mechanism~\ref{alg:budgetedMatroids} will not stop until $k$ is considered. Note that the declared costs are not involved in computing maximum-value independent sets. The maximum-value independent sets computed are exactly the same as declaring truthfully until the $h^{th}$ iteration. Moreover, the maximum-value independent set is also same in the $h^{th}$ iteration since the remaining elements in the matroid are the same.
	As $r$ in the $h^{th}$ is equal to or greater than before, the maximum-value independent set is not budget feasible. It implies that element $k$ will be removed from the matroid and it will never be included in an independent set in Mechanism~\ref{alg:budgetedMatroids}. Therefore, its utility is zero.
	
	We use  similar arguments to show that element $k$
	can not improve his utility by declaring a smaller cost. Consider  that element $k$ declares a smaller cost $d_k < c_k$ and it is considered at the $h^{th}$ iteration where $h \geq k$. Let us focus on how Mechanism~\ref{alg:budgetedMatroids} performs. Until the $k^{th}$ iteration, maximum-value independent sets are the same as $k$ declaring its cost truthfully and they are not budget feasible. Next, the maximum-value independent set in $k^{th}$ is also the same as before. If the independent set is budget feasible, then we know that $\frac{b}{w(\mathsf{MAX}(\mathcal{M}\setminus (T\cup\tau)))}$ is strictly less than $bb(k)=\frac{w_k}{c_k}$.  It is because that the mechanism does not terminate at $r = \frac{w_k}{c_k}$ when element $k$ declares truthfully. Therefore, even element $k$ is in this independent set, the payment will be strictly less than his true cost. On the other hand, if the independent set is not budget feasible, then the mechanism will update its upper bound of payment for each contribution in the next iteration. The new upper bound is  at most $bb(k)=\frac{w_k}{c_k}$. It implies element $k$ will never get a payment greater than his true cost.
\qed
\end{proof}
In the second case we show that the remaining elements that are not included in the independent set cannot benefit by manipulating their costs.

\begin{lemma}
	\label{lem:truthfulness2}
		Assume an element $k$ is in $E - \tau - T - \mathsf{MAX}(\mathcal{M}\setminus (T \cup \tau))$ when it declares its cost truthfully. Then, element $k$ could not improve his utility by declaring a cost $d_k \neq c_k$.
\end{lemma}
\begin{proof}
	As $k \notin  \mathsf{MAX}(\mathcal{M}\setminus (T \cup \tau))$, we know that the utility of element $k$ is zero. Suppose that the mechanism terminates at the $h^{th}$ round  when element $k$ declares its cost truthfully.
	Consider that element $k$ declares a higher cost $d_k > c_k$ and it is consider at the $l^{th}$ iteration where $l < h$. In this case,  Mechanism~\ref{alg:budgetedMatroids} will not stop before or at the $l^{th}$ iteration since the maximum-value independent sets computed are exactly the same as $k$ declaring its cost truthfully and they are not budget feasible. It implies that
	element $k$ will be removed from the matroid and
	it will never be included in an independent set in Mechanism~\ref{alg:budgetedMatroids}
	by declaring a larger cost. Therefore, its utility is still zero.
	
	Secondly, consider that element $k$ declares a  cost $d_k \neq c_k$ and it is considered at the $h^{th}$ iteration. In this case, Mechanism~\ref{alg:budgetedMatroids} will not stop before the $h^{th}$ iteration because the maximum-value independent sets are not feasible.
	The maximum-value independent in the $h^{th}$ iteration is the same as $k$ declaring its cost truthfully since the remaining elements are the same. Therefore, if the independent set is budget feasible, the mechanism will compute the same independent set and payments. Otherwise, element $k$ will be removed as it must be the element with the $h^{th}$ largest buck-per-bang rate. Element $k$ cannot benefit in any case.
	
	Finally, consider that element $k$ declares a  cost $d_k \neq c_k$ and it is considered after  the $h^{th}$ iteration. In this case, Mechanism~\ref{alg:budgetedMatroids}  will compute the same independent set and payments.
\qed
\end{proof}

By similar arguments, we show that elements in the independent set cannot benefit by manipulating their costs.
\begin{lemma}
	\label{lem:truthfulnessInMatching}
		Assume an element $k$ is in $ \mathsf{MAX}(\mathcal{M}\setminus (T \cup \tau))$ when it declares its cost truthfully. Then, element $k$ could not improve his utility by declaring a cost $d_k \neq c_k$.
\end{lemma}
\begin{proof}
	The proof is exactly the same as Lemma~\ref{lem:truthfulness2}.
	Suppose that the mechanism terminates at the $h^{th}$ round  when element $k$ declares its cost truthfully.
	As $k \in \mathsf{MAX}(\mathcal{M}\setminus (T \cup \tau))$, we know the payment of element $k$ is $r \cdot w_k$. Hence, his utility is $r \cdot w_k - c_k$.
	Consider that element $k$ declares a higher cost $d_k > c_k$ and it is considered at the $l^{th}$ iteration where $l < h$. Similar to Lemma~\ref{lem:truthfulness2},  in this case Mechanism~\ref{alg:budgetedMatroids} will not stop before or at the $l^{th}$ iteration since the maximum-value independent sets computed are exactly the same and they are not budget feasible. It implies that
	element $k$ will be removed from the matroid and
	it will never be included in an independent set in Mechanism~\ref{alg:budgetedMatroids}. Therefore, its utility becomes zero.
	
	Secondly, consider that element $k$ declares a  cost $d_k \neq c_k$ and it is considered at the $h^{th}$ iteration. In this case, Mechanism~\ref{alg:budgetedMatroids} will not stop before the $h^{th}$ iteration because the maximum-value independent sets are not feasible.
	The maximum-value independent in the $h^{th}$ round is the same as $k$ declaring its cost truthfully since the remaining elements are the same. Therefore, if the independent set is budget feasible, the mechanism will compute the same independent set and payments. Otherwise, element $k$ will be removed. Element $k$ cannot benefit in any case.
	
	Finally, consider that element $k$ declares a  cost $d_k \neq c_k$ and it is considered after  the $h^{th}$ iteration. In this case, Mechanism~\ref{alg:budgetedMatroids}  will compute the same independent set and payments.
\qed
\end{proof}

\subsection{Individual Rationality}
When Mechanism~\ref{alg:budgetedMatroids} returns $\tau$, the utility of $\tau$ is non-negative as $c_\tau$ is at most $b$. The utilities for other edges are zero. When Mechanism~\ref{alg:budgetedMatroids} returns $\mathsf{MAX}(\mathcal{M} \setminus (T\cup \tau))$, for any element $e \in \mathsf{MAX}(\mathcal{M} \setminus (T\cup \tau))$, that is, $f_e=1$, its utility is $r \cdot w_e- c_e$ which is non-negative since $r \geq \mathsf{bb}(e)$. For other edges, their utilities are zero.

\subsection{Budget Feasibility}
When Mechanism~\ref{alg:budgetedMatroids} returns $\tau$, it only pays $b$ to edge $\tau$. Hence, it is budget feasible. On the other hand, when Mechanism~\ref{alg:budgetedMatroids} returns $\mathsf{MAX}(\mathcal{M} \setminus (T\cup \tau))$, $r$ is used as payment per contribution. As $r = \min\{ \frac{b}{w(\mathsf{MAX}(\mathcal{M} \setminus (T\cup \tau)))},  \mathsf{bb}(i-1)\}$, it guarantees the budget feasibility.

\subsection{Remarks}
In Mechanism~\ref{alg:budgetedMatroids}, we iteratively compute the maximum-value independent set (e.g. Line 4). In the case that the maximum-value independent set is not unique, we assume there is a deterministic tie-breaking rule. Note that all the results still hold under this assumption. For example, the truthfulness of the mechanism will not be compromised since the the maximum-value independent set only consider the weights of the elements that is the public knowledge.

\section{Mechanisms for matroid intersections}\label{sec_Mechanisms for matroid intersections}
In this section we extend our mechanism to matroid intersections. The matroid intersection problem (i.e., finding the maximum-value common independent set) is NP-hard in general when more than three matroids are involved. Some interesting cases of matroid intersection problems can be solved efficiently (i.e., they can be formulated as the intersection of two matroids), for example, matchings in bipartite graphs, arborescences in directed graphs, spanning forests in undirected graphs, etc. Nevertheless we point out that a very similar mechanism to the one presented in last section achieves a $4$ approximation for the case when, instead of a matroid, we are given an undirected weighted (general) graph where the selfish agents are the edges of the graph and the buyer wants to procure a matching under the given budget that yields the largest value possible to him. 

For general matroid intersections, our main result is the following. Given a deterministic polynomial time blackbox $\mathsf{APX}$ that achieves an $\alpha$-approximation to $k$-matroid intersection problems, we provide a polynomial time, individually rational, truthful and budget feasible deterministic  mechanism that is $(3\alpha+1)$-competitive against the maximum-value common independent set. The mechanism is similar to Mechanism~\ref{alg:budgetedMatroids} by changing $\mathsf{MAX}$ to $\mathsf{APX}$. It is well-known that the VCG payment rule does not preserve the property of truthfulness in the presence of approximated solutions (i.e., non-optimal outcome). However unlike the VCG mechanism, we show that Mechanism~\ref{alg:matroidsIntersection} preserves its truthfulness when $\mathsf{APX}$ is used. We believe that this result will make our contribution more practical. 

\begin{algorithm}[h]
	\KwIn{$\mathcal{M}=(E,\mathcal{I}), \mathbf{w}, \mathbf{d},b$}
	\KwOut{$\mathbf{f},\mathbf{p}$}
	Sort elements in $E - \tau$ in a non-increasing order of buck per bang, i.e. $\mathsf{bb}(i) \geq \mathsf{bb}(j)$ if $ i < j$, break ties arbitrarily\;
	Let $\mathsf{bb}(0) = +\infty$, $i=1$ and $T = \emptyset$\; 
	Set $r = \mathsf{bb}(i)$\;
	\While{$w(\mathsf{APX}(\mathcal{M} \setminus T))\cdot r > b$}{
		$T = T \cup \{i\}$ and $i = i+1$\;
	}
	$r = \min\{ \frac{b}{w(\mathsf{APX}(\mathcal{M} \setminus T))},  \mathsf{bb}(i-1)\}$\label{line:HatRBudgetmatroidintersection}\;
	\uIf{$w(\mathsf{APX}(\mathcal{M} \setminus T))>w_\tau$}{
		For each $e \in E$, if $e \in \mathsf{APX}(\mathcal{M} \setminus T), f_e = 1$ and $ p_e = r \cdot w_k$. Otherwise, $f_e = 0$ and $ p_e = 0$\;
	}\Else{
	$f_\tau = 1, p_\tau = b$. For edge $e \in E - \tau, f_e = 0, p_e = 0$\;
}
\KwRet{$\mathbf{f}, \mathbf{p}$}\;
\SetAlgorithmName{Mechanism}{mechanism}{List of Mechanisms}
\caption{A budget feasible mechanism for procuring independent sets in matroid intersections}
\label{alg:matroidsIntersection}
\end{algorithm}

\subsection{Matroid intersections}
Given $k$-matroid $\mathcal{M}_1, \ldots, \mathcal{M}_k$, let $\mathcal{M} = (E, \mathcal{I} )$ be the``true matroid" where $E$ is the common ground elements and $\mathcal{I} = \bigcap_{j} \mathcal{I}_j$ is the ``true independent sets". 
Similar as the notations we used before, let $\mathsf{OPT}(\mathcal{M \setminus T},b)$ and $\mathsf{OPT}(\mathcal{M | T},b)$ denote the optimal independent set satisfying the budget constraint in matroid $\mathcal{M}\setminus T$ and $\mathcal{M} | T$, respectively. Let $\mathsf{APX}(\mathcal{M \setminus T},b)$ be the maximum-value independent set in matroid  $\mathcal{M} \setminus T$ returned by the $\alpha$-approximation algorithm.

\subsection{Obtaining $O(\alpha)$ approximation}
We show the following key lemma, which is similar to Lemma~\ref{lem:approximaiton} and implies the approximation of our mechanism for matroid intersections.

\begin{lemma}
	\label{lem:matroidintersection}
	Given any $\mathcal{M}, \mathbf{w}, \mathbf{d}, b$, when Mechanism~\ref{alg:matroidsIntersection} stops, it holds
	\[
	w(\mathsf{OPT}(\mathcal{M} \setminus  \tau, b)) \leq 2 \cdot \alpha \cdot w(\mathsf{APX}(\mathcal{M}\setminus (T \cup \tau) )) + \alpha \cdot w_\tau
	\]
\end{lemma}
\begin{proof}
	The proof has the same spirit as the proof of  Lemma~\ref{lem:approximaiton}. We consider two cases depending on whether the full budget $b$ is spent or not. 	Consider $E - \{\tau\} $ is partitioned into two disjoint sets, $E - \{\tau\} - T$ and ${T}$.
	Similar to Lemma~\ref{lem:approximaiton}, when the full budget is spent, we get
	\begin{align*}
	 w(\mathsf{OPT}(\mathcal{M} \setminus  \tau, b))  
	\leq &  w(\mathsf{OPT}(\mathcal{M} | T, b)) +  w(\mathsf{OPT}(\mathcal{M} \setminus (T \cup \tau)    ,b)) \\
	\leq & \frac{b}{r} + \alpha \cdot w(\mathsf{APX}(\mathcal{M}\setminus (T \cup \tau) )) \\
	\leq & (\alpha + 1) \cdot  w(\mathsf{APX}(\mathcal{M}\setminus (T \cup \tau) ))
	\end{align*}
	When there is some budget left in Mechanism~\ref{alg:matroidsIntersection}, the analysis involves one more step compared to Lemma~\ref{lem:approximaiton} although the idea is still to bound the budget left. Since Mechanism~\ref{alg:matroidsIntersection} does not stop when $r =\mathsf{bb}(i-1)$, it implies that the independent set returned by $\mathsf{APX}$ was not budget feasible at previous iteration. It further implies that the maximum-value independent set is not budget feasible either if the payment per weight is $r$. After removing element $i-1$, the independent set returned by $\mathsf{APX}$ becomes budget feasible when $r= \mathsf{bb}(e_{i-1})$. These together imply
	\[
	w(\mathsf{MAX}(\mathcal{M} \setminus (T' \cup \tau) )) \cdot  \mathsf{bb}(i-1)  \geq 	w(\mathsf{APX}(\mathcal{M} \setminus (T' \cup \tau) )) \cdot  \mathsf{bb}(i-1) > b > w(\mathsf{APX}(\mathcal{M} \setminus (T \cup \tau) ))  \cdot \mathsf{bb}(i-1)
	\]
	where $T' = T - \{i-1\}$. As the sum of $w(\mathsf{APX}(\mathcal{M} \setminus (T \cup \tau) ))$ and $w(i-1)$ is at least $\frac{1}{\alpha}$ fraction of $w(\mathsf{MAX}(\mathcal{M} \setminus (T' \cup \tau) ))$, we get
	\[
	\big( w(\mathsf{APX}(\mathcal{M} \setminus (T \cup \tau) ) + w_{i-1} \big) \cdot \mathsf{bb}(i-1)  
	\geq  \frac{1}{\alpha} \cdot w(\mathsf{MAX}(\mathcal{M} \setminus (T' \cup \tau) )) \cdot \mathsf{bb}(i-1)  
	>  \frac{b}{\alpha}
	\]
	Hence, we get $w(\mathsf{APX}(\mathcal{M} \setminus (T \cup \tau) )  +  w_{i-1}> \frac{b}{\alpha \cdot \mathsf{bb}(i-1)} $. Finally, 
	\begin{align*}
	\mathsf{OPT}(\mathcal{M} \setminus \tau, b) 
	\leq &  w(\mathsf{OPT}(\mathcal{M} | T, b)) +  w(\mathsf{OPT}(\mathcal{M} \setminus (T \cup \tau)    ,b)) \\
	\leq &\frac{b}{\mathsf{bb}(i-1)} + \alpha \cdot w(\mathsf{APX}(\mathcal{M} \setminus (T \cup \tau) \\
	\leq & 2 \cdot \alpha \cdot w(\mathsf{APX}(\mathcal{M} \setminus (T \cup \tau)  + \alpha \cdot w_{i-1}
	\end{align*}
	Substituting $w_{i-1}$ with $w_\tau$ completes the proof.
\qed
\end{proof}

Now, we show the competitive ratio of Mechanism~\ref{alg:matroidsIntersection}

\begin{lemma}
	Given any $\mathcal{M}, \mathbf{w}, \mathbf{d}, b$, the independent set returned by Mechanism~\ref{alg:matroidsIntersection}, i.e., the maximum between $\tau$ or $\mathsf{APX}(\mathcal{M}\setminus (T \cup \tau) )$, is $4\alpha-$competitive against the optimal independent set.
\end{lemma}
\begin{proof}
	The optimal independent set in $\mathcal{M}$ is bounded by
	\[
	w(\mathsf{OPT}(\mathcal{M},b)) \leq w_\tau + w(\mathsf{OPT}(\mathcal{M} \setminus  \tau, b))
	\]
	By Lemma~\ref{lem:approximaiton}, we have
	\[
	w(\mathsf{OPT}(\mathcal{M},b)) \leq (\alpha+1)w_\tau + 2\cdot \alpha \cdot w(\mathsf{APX}(\mathcal{M}\setminus (T \cup \tau) ))
	\]
	Therefore, the maximum between $\tau$ and $\mathsf{APX}(\mathcal{M}\setminus (T \cup \tau) )$ approximates the optimal independent set within a factor of $3\alpha+1$.
\qed
\end{proof}

\subsection{Preserving the truthfulness}
In this section, we will show that replacing $\mathsf{MAX}$ by $\mathsf{APX}$ preserve the truthfulness of the mechanism for matroid intersections. The reason behind is that the mechanism works in a greedy fashion and at each iteration the cost declared by elements \textit{does not} affect the independent set computed in the mechanism. The property of the truthfulness replies on the greedy approach instead of the optimality of the independent set. Informally speaking, if an element declares a cost rather than its true cost, its utility will remain the same or it will get removed. The proofs are similar to the proofs in Section~\ref{sec:truthfulmatroid}. 

\begin{lemma}
	\label{lem:truthfulmatroidintersection1}
	Assume an element $k \in T$ when it declares its cost truthfully. Then, element $k$ could not improve his utility by declaring a cost $d_k \neq c_k$.
\end{lemma}
\begin{proof}
	As $k \in T$, we know that element $k$ is not in the independent set returned by Mechanism~\ref{alg:matroidsIntersection}. Hence, his utility is zero. It
	implies, when element $k$ is considered, that is, $r = \mathsf{bb}(k)$, we get $w(\mathsf{APX}(\mathcal{M} \setminus (T^{k} \cup \tau)  )  )) \cdot r > b$ where $T^k$ denotes the set of elements removed until $k$ is considered. Consider that element $k$ declares a higher cost $d_k > c_k$ and it is considered earlier at the $h^{th}$ iteration where $h\leq k$. Equivalently speaking,  $k$ becomes the element with the $h^{th}$ largest buck-per-bang rate.  In this case,  Mechanism~\ref{alg:matroidsIntersection} will not stop until $k$ is considered as the independent sets computed in $\mathsf{APX}$ are the same as $k$ declaring its cost truthfully. Moreover, the independent set is also the same in the $h^{th}$ iteration as the remaining elements are the same. 
	As $r$ in the $h^{th}$ is equal to or greater than before, the independent set is not budget feasible. It implies that element $k$ will be removed from the matroid. 
	It concludes that
	element $k$ will never be included in an independent set in Mechanism~\ref{alg:budgetedMatroids}. Therefore, its utility is still zero. 
	
	On the other hand, consider  that element $k$ declares a smaller cost $d_k < c_k$ and it is considered at the $h^{th}$ iteration where $h \geq k$. Let us focus on how Mechanism~\ref{alg:matroidsIntersection} performs. Until the $k^{th}$ iteration, the independent sets computed in $\mathsf{APX}$ are the same as $k$ declaring its cost truthfully and they are not budget feasible. Next, the independent set in the $k^{th}$ iteration is the same as before. If the independent set is budget feasible, then we know that $r$ is strictly less than $bb(k)=\frac{w_k}{c_k}$.  It is because that the mechanism does not terminate at $r = \frac{w_k}{c_k}$ when element $k$ declares truthfully. Therefore, even element $k$ is in this independent set, the payment will be strictly less than his true cost. On the other hand, if the independent set is not budget feasible, the mechanism will update its upper bound of payment. The new upper bound is at most $bb(k)=\frac{w_k}{c_k}$. It implies element $k$ will never get a payment greater than his true cost. 
\qed
\end{proof}

\begin{lemma}
	\label{lem:truthfulness2matoidintersection}
	Assume an element $k$ is in $E - \tau - T - \mathsf{APX}(\mathcal{M}\setminus (T \cup \tau))$ when it declares its cost truthfully. Then, element $k$ could not improve his utility by declaring a cost $d_k \neq c_k$.
\end{lemma}
\begin{proof}
	As $k \notin  \mathsf{APX}(\mathcal{M}\setminus (T \cup \tau))$, we know that the utility of element $k$ is zero. Suppose that the mechanism terminates at the $h^{th}$ round  when element $k$ declares its cost truthfully.  
	Consider that element $k$ declares a higher cost $d_k > c_k$ and it is consider at the $l^{th}$ iteration where $l < h$. In this case,  Mechanism~\ref{alg:matroidsIntersection} will not stop before or at the $l^{th}$ iteration since the independent sets computed in $\mathsf{APX}$ are exactly the same as $k$ declaring its cost truthfully and they are not budget feasible. It implies that
	element $k$ will be removed from the matroid and 
	it will never be included in an independent set in Mechanism~\ref{alg:matroidsIntersection}. Therefore, its utility is still zero. 
	
	Secondly, consider that element $k$ declares a  cost $d_k \neq c_k$ and it is considered at the $h^{th}$ iteration. In this case, Mechanism~\ref{alg:matroidsIntersection} will not stop before the $h^{th}$ iteration because the  independent sets are not feasible. 
	The independent in the $h^{th}$ iteration is the same as $k$ declaring its cost truthfully since the remaining elements are the same. Therefore, if the independent set is budget feasible, the mechanism will compute the same independent set and payments. Otherwise, element $k$ will be removed as it must be the element with the $h^{th}$ largest buck-per-bang rate. Element $k$ cannot benefit in any case. 
	
	Finally, consider that element $k$ declares a  cost $d_k \neq c_k$ and it is considered after  the $h^{th}$ iteration. In this case, Mechanism~\ref{alg:matroidsIntersection}  will compute the same independent set and payments. 
\qed
\end{proof}

\begin{lemma}
	The element with the maximum weight, i.e., element $\tau$, could not improve his utility by declaring a cost $d_\tau \neq c_\tau$.
\end{lemma}
\begin{proof}
	If Mechanism~\ref{alg:matroidsIntersection} returns element $\tau$ when  $\tau$ declares his true cost, then $\tau$ gets a payment of $b$ so that there is no incentive for him to declare other cost. On the other hand, when Mechanism~\ref{alg:matroidsIntersection} returns $\mathsf{APX}(\mathcal{M}\setminus (T\cup \tau))$, declaring a different cost will not change the outcome as it is still the element with the largest weight and $w_\tau < w(\mathsf{APX}(\mathcal{M}\setminus (T\cup \tau)))$.
\qed
\end{proof}

\begin{lemma}
	Assume an element $k$ is in $ \mathsf{APX}(\mathcal{M}\setminus (T \cup \tau))$ when it declares its cost truthfully. Then, element $k$ could not improve his utility by declaring a cost $d_k \neq c_k$.
\end{lemma}
\begin{proof}
	The proof is exactly the same as Lemma~\ref{lem:truthfulness2matoidintersection}.
	Suppose that the mechanism terminates at the $h^{th}$ round  when element $k$ declares its cost truthfully.
	As $k \in \mathsf{APX}(\mathcal{M}\setminus (T \cup \tau))$, we know the payment of element $k$ is $r \cdot w_k$. Hence, his utility is $r \cdot w_k - c_k$. 
	Consider that element $k$ declares a higher cost $d_k > c_k$ and it is considered at the $l^{th}$ iteration where $l < h$. Similar to Lemma~\ref{lem:truthfulness2matoidintersection},  in this case Mechanism~\ref{alg:matroidsIntersection} will not stop before or at the $l^{th}$ iteration since the independent sets computed in $\mathsf{APX}$ are exactly the same as $k$ declaring its cost truthfully and they are not budget feasible. It implies that
	element $k$ will be removed from the matroid and 
	it will never be included in an independent set in Mechanism~\ref{alg:matroidsIntersection}. Therefore, its utility becomes zero. 
	
	Secondly, consider that element $k$ declares a  cost $d_k \neq c_k$ and it is considered at the $h^{th}$ iteration. In this case, Mechanism~\ref{alg:matroidsIntersection} will not stop before the $h^{th}$ iteration because the maximum-value independent sets are not feasible. 
	The  independent in the $h^{th}$ round is the same as $k$ declaring its cost truthfully since the remaining elements are the same. Therefore, if the independent set is budget feasible, the mechanism will compute the same independent set and payments. Otherwise, element $k$ will be removed. Element $k$ cannot benefit in any case. 
	
	Finally, consider that element $k$ declares a  cost $d_k \neq c_k$ and it is considered after  the $h^{th}$ iteration. In this case, Mechanism~\ref{alg:matroidsIntersection}  will compute the same independent set and payments. 
\qed
\end{proof}

\section{Applications}
In this section we briefly discuss some applications of our results. 

\noindent{\bf Uniform Matroid} Additive valuation has been studied in the design of budget feasible mechanisms, e.g.~\cite{singer2010budget,chen2011approximability}. In such settings a buyer would like to maximize his valuation by procuring items under the constraint that his payment is at most his budget. Our result generalizes to the case where the buyer has not only the budget constraint but also has a limit on the number of items he can buy. For example hiring people in companies is not only constraint by budgets but also limited by the office space.  

\noindent{\bf Scheduling Matroid} Our mechanism could be used to purchase processing time in the context of job scheduling. One special case is the following. Each job is associated with a deadline and a profit, and requires a unit of processing time. As jobs may conflict with each other, only one job can be scheduled at the same time. The buyer would like to maximize his profit by completing jobs under the constraint that he does not spend more than his budget in purchasing processing time.  

\noindent{\bf Spectrum Market} Tse and Hanly~\cite{tse1998multiaccess} showed that the set of achievable rates in a Gaussian multiple-access, known as the Cover-Wyner capacity region, forms a polymotroid. It is known there is a pseudopolynomial reduction from polymatroids to matroids~\cite{schrijver2003combinatorial}. Therefore, our mechanism can be used to purchase transmission rates by tele-communication companies.  
\section{XOS functions}\label{sec_XOS functions}
Budget feasible mechanisms received a lot of attention also when the valuations functions are submodular~\cite{singer2010budget} and XOS~\cite{bei2012budget}. In this study, we also slightly improve the analysis of the mechanism proposed in~\cite{bei2012budget}. Specifically, we improve the approximation of the mechanism from $768$ to $436$ by tuning the parameters in the mechanism. 

\begin{theorem}
	There exists a randomized universally truthful mechanisms that provides a $436$-approximation ratio for XOS valuation functions. 
\end{theorem}

\subsection{Model}
We are given a set $E$ consisting of $n$ elements and a budget $b$. Each element $e \in E$ has a private cost $c_e$. For any subset $S \subseteq E$, there is publicly known valuation function $v(S)$ that indicates the value of $S$. In this section, we are interested in XOS valuation functions. More precisely, a function $v(\cdot)$ is XOS if 
\[
v(S) = \max \{ f_1(S),f_2(S),\ldots,f_m(S)   \}  \hspace{1cm} \textit{ for any } S\subseteq E
\]	
where each $f_k(\cdot)$ is a nonnegative additive function.

Our goal is to design truthful mechanisms that give elements incentives to declare their true private costs. Meanwhile, mechanisms aim to choose a set of elements within the budget and maximize the value of the chosen agents.
Unlike matroids, the mechanism is allowed to select any subset of elements. 
We compare our mechanisms against the optimal mechanism which alway magically knows the private costs of elements. The optimal mechanism returns a set of elements such that the aggregated cost of elements is at most budget $b$ and the value of the elements is maximized.  We compare the value of elements chosen by our mechanisms against the value of elements returned by the optimal mechanism.  

Similar as truthful mechanisms, when a mechanism is randomized, that is, outputting a distribution over a set of outcomes, we call a randomized mechanism  universally truthful if it takes a distribution over deterministic truthful mechanisms.

In the mechanism and its analysis presented in the following sections, we will often consider the optimal solution in a restricted set of elements. Given $S \subseteq E$, let $\mathrm{OPT}(S)$ be the optimal solution when only elements in $S$ are the input of the problem. For example, the optimal solution of the problem is denoted by $\mathrm{OPT}(E)$. We will simply use $\mathrm{OPT}$ to denote the optimal solution of the problem, i.e., $\mathrm{OPT} = \mathrm{OPT}(E)$.
Let $f^*$ be the additive function in the XOS definition of $v(\cdot)$ with $f^*(\mathrm{OPT})=v(\mathrm{OPT})$.

\subsection{Key Lemmas}
\begin{lemma}
	\label{lem:partition}
	Assume that  $f^*(e) \leq \frac{1}{\alpha}f^*(\mathrm{OPT})$   for all $e \in \mathrm{OPT}$, then there exists two disjoint sets $S_1, S_2 \subset E$ such that $v(S_1) \geq \frac{\alpha-1}{2\alpha} f^*(\mathrm{OPT})$ and $v(S_2) \geq \frac{\alpha-1}{2\alpha} f^*(\mathrm{OPT})$.
\end{lemma}
\begin{proof}
	We give a constructive proof for this lemma. In the next paragraph, we will show a way to construct sets $S_1$ and $S_2$ such that $f^*(S_1) \geq \frac{\alpha-1}{2\alpha} f^*(\mathrm{OPT})$ and $f^*(S_2) \geq \frac{\alpha-1}{2\alpha} f^*(\mathrm{OPT})$. Given that $v(S_1) \geq f^*(S_1)$ and $v(S_2) \geq f^*(S_2)$ implied by the definition of XOS functions, the lemma directly follows. 
	
	Consider an arbitrary order of elements in $\mathrm{OPT}$, we keep adding elements to $S_1$ until that $f^*(S_1) \geq \frac{\alpha-1}{2\alpha} f_{OPT}(\mathrm{OPT})$. Then, the rest of agents are included in $S_2$. By the assume that  $f^*(e) \leq \frac{1}{\alpha}f^*(\mathrm{OPT})$   for all $e \in \mathrm{OPT}$, it implies that $f^*(S_1) \leq \frac{\alpha+1}{2\alpha}f^*(\mathrm{OPT})$. Finally, since $f^*(S_1) + f^*(S_2) = f^*(\mathrm{OPT})$, we have $f^*(S_2) \geq \frac{\alpha-1}{2\alpha} f^*(\mathrm{OPT})$.
\end{proof}

Next, we show that by partitioning elements uniformly random into two groups, we can have a good approximation to the optimal solution at both groups in expectation. The proof shares the same spirit as Lemma~2.1 in~\cite{bei2012budget}.
\begin{lemma}
	\label{lem:randomSampling}
	Assume that $f^*(e) \leq \frac{1}{\alpha}f_{OPT}(\mathrm{OPT})$ for all $e \in \mathrm{OPT}$. Furthermore, suppose that $E$ is divided uniformly at random into two groups $T_1$ and $T_2$. Then, with probability of at least $\frac{1}{2}$, it holds that $v(T_1) \geq \frac{\alpha-1}{4\alpha}  f^*(\mathrm{OPT})$ and $v(T_2)\geq \frac{\alpha-1}{4\alpha} f^*(\mathrm{OPT})$. 
\end{lemma}
\begin{proof}
	Let $S_1, S_2$ be two disjoint sets such that $f^*(S_1) \geq \frac{\alpha-1}{2\alpha} f^*(\mathrm{OPT})$ and $f^*(S_2) \geq \frac{\alpha-1}{2\alpha} f^*(\mathrm{OPT})$. Consider $X_1 = S_1 \cap T_1, Y_1 = S_2 \cap T_1, X_2 = S_1 \cap T_2, Y_2 = S_2 \cap T_2$. As partitioning $S_1$ into
	$X_1,Y_1$ and partitioning $S_2$ into $X_2, Y_2$ are independent to each other. Therefore, with probability $\frac{1}{2}$, the most valuable parts of $S_1$ and $S_2$ will get into different sets $T_1$ and $T_2$, respectively. Thus the lemma follows.
\end{proof}

\subsection{Mechanism}

\fbox{\parbox{\textwidth}{\textsc{XOS-MECHANISM-MAIN}$(\alpha, \beta)$:
		\begin{enumerate}
			\item W.p. $\frac{1}{2}$, pick the most value element and pay him $b$. W.p. $\frac{1}{2}$, continue.
			\item Divide elements independently at random with probability $\frac{1}{2}$ into two set $T_1$ and $T_2$. 
			\item Compute an optimal solution $\mathrm{OPT}(T_1)$ for elements in $T_1$ given budget $b$. 
			\item Set a threshold $t = \frac{v(\mathrm{OPT}(T_1))}{\beta \cdot b}$.
			\item Find a set $S^* \subseteq T_2$ such that
			\[
			S^* \in \arg max_{S \subseteq T_2} \{v(S) - t\cdot c(S) \} 
			\]
			where   $c(S) = \sum_{e\in S} c_e$.
			\item Let $f$ be the additive function in the XOS definition of $v(\cdot)$ with $f(S^*) = v(S^*)$.
			\item Run \textsc{ADDITIVE-MECHANISM} for  $f$ with respect to set $S^*$ and budget $b$.
			\item Output the result of \textsc{ADDITIVE-MECHANISM}.
		\end{enumerate}}}
		
		\begin{lemma}[Claim 3.1 in~\cite{bei2012budget}]
			\label{lem:subset}
			For any $S \subseteq S^*, f(S) - t \cdot c(S) \geq 0$. 
		\end{lemma}

		\begin{lemma}
			\label{lem:constantApprox}
			\textsc{XOS-MECHANISM-MAIN} has a $436$-approximation ratio.
		\end{lemma}
		\begin{proof}
			We prove this lemma by considering difference cases. First, assuming that there exists an element $e$ such that $f^*(e) > \frac{1}{\alpha} f^*(\mathrm{OPT})$, \textsc{XOS-RANDOM-SAMPLE} has a probability of $\frac{1}{2}$ to return the most value agent. Hence, \textsc{XOS-RANDOM-SAMPLE} is $2\alpha$-approximate the optimal in this case. 
			
			Second, we consider the case that $f^*(e) \leq \frac{1}{\alpha} f^*(\mathrm{OPT})$ for all $e \in \mathrm{OPT}$. The main idea is to show that there exists a $S' \subseteq S^*$ such that  $c(S')$ is at most $b$ and $f(S')$ is a good approximation to $f^*(\mathrm{OPT})$. Let us divide this case into two sub-cases.
			\begin{itemize}
				\item $c(S^*)>b$. In this case, since $c(e) \leq b$ for all $e \in E$, we can always find a subset $S' \subset S^*$ such that $\frac{b}{2} \leq c(S') \leq b$.  By Lemma~\ref{lem:subset}, we know $f(S') \geq t \cdot c(S') \geq \frac{v(\mathrm{OPT}(T_1))}{\beta \cdot b} \cdot \frac{b}{2} \geq \frac{v(\mathrm{OPT}(T_1))}{2\beta}$. As $f(\mathrm{OPT}(S^*))$ is at least $f(S')$, we have $f(\mathrm{OPT}(S^*)) \geq f(S') \geq  \frac{v(\mathrm{OPT}(T_1))}{2\beta} \geq \frac{\alpha-1}{8\cdot\alpha \cdot \beta}f^*(\mathrm{OPT})$ with a probability of at least $\frac{1}{2}$. 
				\item $c(S^*)\leq b$. Then $\mathrm{OPT}(S^*) = S^*$. Let $S' = \mathrm{OPT} \setminus T_1$, thus, $c(S') \leq c(\mathrm{OPT}) \leq b$. By Lemma~\ref{lem:randomSampling}, we have $v(S') \geq \frac{\alpha-1}{4\alpha}f^*(\mathrm{OPT})$ with a probability of at least $\frac{1}{2}$. Since $S^* \in \arg max_{S \subseteq T_2} \{v(S) - t\cdot c(S) \}$, with a probability of at least $\frac{1}{2}$, we have
				\begin{align*}
				f(\mathrm{OPT}(S^*)) = f(S^*) & = v(S^*) 
				\\ & \geq v(S^*) - t \cdot c(S^*) 
				\\ & \geq v(S') - t \cdot c(S') 
				\\ & \geq \frac{\alpha-1}{4\alpha} f^*(\mathrm{OPT}) -  \frac{v(\mathrm{OPT}(T_1))}{\beta \cdot b} \cdot b 
				\\ & \geq \frac{\alpha-1}{4\alpha} f^*(\mathrm{OPT}) -  \frac{f^*(\mathrm{OPT})}{\beta} 
				\\ &  = \frac{\alpha \cdot \beta - \beta - 4\alpha}{4 \cdot \alpha \cdot \beta}f^*(\mathrm{OPT})
				\end{align*}
			\end{itemize}
			In both cases, we run \textsc{ADDITIVE-MECHANISM} which has an approximation factor of $3$ to $f(\mathrm{OPT}(S^*))$. Therefore, the approximation ratio for the case that $f^*(e) \leq \frac{1}{\alpha} f^*(\mathrm{OPT})$ for all $e \in \mathrm{OPT}$ is 
			\[
			\frac{1}{2} \cdot \frac{1}{3} \cdot \frac{1}{2} \min (\frac{\alpha-1}{8\cdot\alpha \cdot \beta},   \frac{\alpha \cdot \beta - \beta - 4\alpha}{4 \cdot \alpha \cdot \beta}  )
			\]
			To combine with the first case, we conclude that the approximation of the mechanism is 
			\[
			\min \big(\frac{1}{2\alpha}, \frac{1}{2} \cdot \frac{1}{3} \cdot \frac{1}{2} \min (\frac{\alpha-1}{8\cdot\alpha \cdot \beta},   \frac{\alpha \cdot \beta - \beta - 4\alpha}{4 \cdot \alpha \cdot \beta}  ) \big)
			\]
			By setting $\alpha \approx 218$ and $\beta \approx 4.5$, we get the approximation of $436$. 
		\end{proof}
		\begin{lemma}[Lemma 3.1 in~\cite{bei2012budget}]
			\label{lem:unversiallytruthful}
			\textsc{XOS-MECHANISM-MAIN} is universally truthful .
		\end{lemma}

\section{Acknowledgement}
This work was partially supported by the ERC Advanced Grant 788893 AMDROMA ``Algorithmic and Mechanism Design Research in Online Markets" and MIUR PRIN project ALGADIMAR ``Algorithms, Games, and Digital Markets".

\bibliographystyle{plain}
\bibliography{ref}

\end{document}